\documentclass[pra,aps,twocolumn,superscriptaddress, longbibliography, notitlepage]{revtex4-1}
\usepackage[colorlinks=true, citecolor=blue, urlcolor=blue ]{hyperref}
\usepackage{graphicx}
\usepackage{dcolumn}
\usepackage{bm}
\usepackage{amsmath,amsfonts}
\usepackage[english]{babel}
\usepackage{amsthm}
\usepackage{hyperref}
\theoremstyle{plain}
\newcommand{\new}[1]{\textcolor{teal}{#1}}

\newtheorem{theorem}{Theorem}

\newtheorem{lemma}[theorem]{Lemma}

\usepackage{xcolor}
\begin{document}


\title{Entanglement witnessing by arbitrarily many independent observers\\recycling a  local quantum shared state}


\author{Chirag Srivastava}
\affiliation{Harish-Chandra Research Institute, HBNI, Chhatnag Road, Jhunsi, Allahabad 211 019, India}
\author{Mahasweta Pandit}
\affiliation{Institute of Theoretical Physics and Astrophysics, Faculty of Mathematics, Physics and Informatics, University of Gdańsk, 80-308 Gdańsk, Poland}
\author{Ujjwal Sen}
\affiliation{Harish-Chandra Research Institute, HBNI, Chhatnag Road, Jhunsi, Allahabad 211 019, India}



\begin{abstract}
We investigate the scenario where an observer, Alice, shares a two-qubit state with an arbitrary number of observers, Bobs, via sequentially and independently recycling the qubit in possession of the first Bob. It is known that there exist entangled states which can be used to have an arbitrarily long sequence of Bobs who can violate the Clauser-Horne-Shimony-Holt (CHSH) Bell inequality with the single Alice. We show that there exist entangled states that do not violate the Bell inequality and whose entanglement can be detected by an arbitrary number of Bobs by suitably choosing the entanglement witness operator and the unsharp measurement settings by the Bobs. This proves that the set of states that can be used to witness entanglement sequentially is larger than those that can witness sequential violation of local realism.
There exist, therefore, two-party quantum correlations that are Bell ``classical'',  but
whose entanglement ``nonclassicality'' can be  
witnessed sequentially and independently by an arbitrarily large number of observers at one end of the shared state with the single observer at the other end.
\end{abstract}

\maketitle

\section{Introduction}
\noindent Quantum entanglement~\cite{horodecki09,Guhne09}
~and 
Bell inequality violation~\cite{Bell64,Brunner14} are among the most prominent non-classical phenomena witnessed in quantum systems.
While entanglement is necessary to exhibit Bell inequality violation, the converse is not true. In \cite{Werner89}, Werner showed that there exist  entangled 
bipartite states, invariant under equal local unitaries, that admit a local hidden variable (LHV) model for arbitrary von Neumann measurements. 
It was later proven~\cite{Barrett02} that such states exist even for the most general non-sequential measurements or the so-called positive operator valued measures (POVMs). 
This provides a natural categorization among entangled states, viz. those that violate a Bell inequality and those that satisfy them. For brevity, and according to the practice prevalent in the literature, we will sometimes refer to them as ``nonlocal'' and ``local'' entangled states. 
It is clear that the ``local'' entangled states are so, within a certain set-up of Bell inequality violation experiments, specifying  e.g. the number of settings per site, the number of outcomes per setting, etc. 

Entangled states that violate a Bell inequality have been found to be 
useful in device-independent tasks involving key distribution \cite{Mayers98,Barrett05,Acin07}, randomness amplification \cite{Colbeck12}, and randomness expansion \cite{Colbeck111,Colbeck112,Pironio10}.
However, local entangled states can also be resourceful, and in particular for two qubits, all entangled states are useful for quantum teleportation beyond the classical limit \cite{Ben93,Hor99,Nielsen02}

In this work, we consider the following scenario \cite{Silva15}:
~we investigate whether given an initial bipartite entangled state, whose one-half is controlled by a single observer, Alice, and the other half by multiple sequential observers, Bobs, can give rise to a sequential violation of the Bell Clauser-Horne-Shimony-Holt (CHSH) inequality~\cite{Clauser69}. If $n$ Bobs observe violations then we call the initial state $n$-recyclable with respect to nonlocal correlations.
It was shown that for an arbitrary $n$, there exist $n$-recyclable states with respect to nonlocal correlations, when their input distributions (probabilities of measurement settings) are modified in such a way that one of the inputs is highly favored, i.e., a ``biased'' measurement strategy is employed ~\cite{Silva15}. 
Later, it was found that such states can be found even for the unbiased case~\cite{Brown20}. 
For 
other works on sequential violation of Bell inequalities, see e.g. \cite{Mal16,Hu18,Saha19,Das19,Vallone20,Fei21,roy20,cabello20,ren21,zhu21,Cheng21,Hall21}.

The measurement strategy employed in \cite{Brown20}, is sufficient for an unbounded number of recycled nonlocal correlations if the first Bob share certain nonlocal quantum states with the Alice. However,  determination of whether all nonlocal quantum states can guarantee generation of unbounded number of such nonlocal correlations remains  an important open problem. The same question can be raised about the set of quantum states shared by the first pair of observers in the context of entangled correlations instead of nonlocal correlations. Since violation of the CHSH inequality also implies the existence of entanglement, the states that are  in the set of $n$-recyclable states with respect to the nonlocal correlations, are also the $n$-recyclable entangled states.\\ 
~In this paper, we take a step further and investigate whether there exist local entangled states that are $n$-recyclable with respect to entangled correlations. We find this problem to be particularly interesting for the following two reasons. Firstly, as we present an affirmative answer to the problem, it is proven that the set of $n$-recyclable entangled states is strictly larger than the set of $n$-recyclable states with respect to the nonlocal correlations, and this holds for arbitrary \(n\). Secondly, since we show that even weakly entangled states can exhibit  robustness and be $n$-recyclable, for arbitrary \(n\), the results of this article will be potentially useful in realistic applications. For previous works on sequential detection of entanglement, see \cite{Bera18,Maity20,sriv21,Vallone20}.\\
Sequential detection of entanglement and finding the length of that sequence tell us about a fundamental aspect of entanglement. They provide us information about whether there exists a fundamental limit on the recyclability of the crucial resource. 
Also, the finding that at least for some entangled states, the sequence is unbounded, leads us to a potentially fresh classification of entangled states between bounded and unbounded sequence states. Currently, there is a possibility that the former set is empty.
Furthermore, the sequential detection scenario provides us with another face of the quantum information gain - state disturbance trade-off:
entanglement detection implies information gain about the system, which comes at a cost of disturbing the state, and 
it is not obvious a priori how long the sequences can be, for specific entangled state inputs. 
Along with these fundamental aspects, sequential detection can potentially be useful for applications in quantum technology. In particular, it could be crucially important in physical substrates where state preparation is costly~\cite{Brown20}.

\section{prerequisites}
\subsection*{Entangled states and local hidden variables
}

In general, any two-qubit quantum state, $\rho$, 
can be expressed as 
\begin{equation} \label{genrho}
\begin{split}
\rho=\frac{1}{4}[\mathbb{I}_2 \otimes \mathbb{I}_2 + \sum_{i=1}^3 m_i \mathbb{I}_2 \otimes \sigma_i &+ \sum_{i=1}^3 n_i \sigma_i \otimes \mathbb{I}_2\\ +& \sum_{i,j=1}^3 t_{ij} \sigma_i \otimes \sigma_j]
\end{split}
\end{equation}
where $\sigma_i \in \{\sigma_1,\sigma_2,\sigma_3\}$ are the Pauli matrices, $\mathbb{I}_2$ is the identity operator on the 
qubit 
Hilbert space, and $m_i$, $n_i$, and $t_{ij}$ are real numbers in $[-1,1]$ 
such that $\rho$ is positive-semidefinite. 

The state $\rho$ is entangled if and only if an eigenvalue of its partial transpose is negative~\cite{Peres96,Hor96}. The matrix $T = (t_{ij})$ transforms like a tensor under local unitary transformations on the qubits and therefore referred to as the correlation tensor. Let $u_0$ and $u_1$ denote the two largest eigenvalues of the matrix $T^\dagger T$. A necessary and sufficient condition~\cite{horodecki95} that $\rho$ will violate a CHSH Bell inequality is the violation of the inequality
\begin{equation}\label{local}
u_0+u_1\leq 1.
\end{equation} 
Once ineq.~(\ref{local}) is satisfied, the state admits a local hidden variable  model for any set of two rank-1 projection-valued measurements per site~\cite{Fine82}. In this paper, we often refer to such states as ones that admit an LHV model. 
\subsection*{Entanglement witnesses}
The concept of entanglement witnesses \new{\cite{Guhne09,horodecki09,terhal00,Sarbicki14}} provides an efficient method to detect entanglement in the laboratory. This method is a consequence of the Hahn-Banach theorem~\cite{Simmons63,Lax02} which says that there always exists a functional on a normed linear space which separates any point in the space from any closed convex set of that space that does not contain the point.
Thus, as the set of separable states form a closed and convex set on the space containing the density matrices, for every entangled state, \(\rho_e\), one can construct an operator $W$, such that
\begin{equation}
    \text{Tr}\{\rho W\}\geq 0, ~~ \forall \rho \in \mathcal{S}
\end{equation}
and 
\begin{equation}
    \text{Tr}\{\rho_e W\}<0,
\end{equation}
where $\mathcal{S}$ is the set of separable states, in the space to which \(\rho\) and \(\rho_e\) belong.
%
Bell inequalities are also 
entanglement witnesses, but 
cannot be used to detect entanglement of states admitting the relevant LHV model.

\section{Scenario}\label{scenario}
We consider the simplest scenario~\cite{Silva15} of sequential sharing of a bipartite entangled quantum state, whose one half is in possession of a single observer, Alice $(A)$, whereas the other half is being measured (and passed on) by $n$ observers, Bobs, with $B_k$ referring to the $k^{\text{th}}$ observer.
These multiple observers on the second system act sequentially and independently on their part of the shared state to detect the shared entanglement with the single observer, Alice, at the other site. The task is to maximize the number of sequential observers who can detect entanglement with the single observer. Since an arbitrary number of independent observers at the second site can share recycled ``nonlocality'' (Bell inequality violation)~\cite{Brown20}, here, we only consider entangled states which admit an LHV model. Therefore, $A$ and $B_1$ start with an entangled state satisfying ineq.~\eqref{local}. $B_1$ performs his measurements and passes his part of the shared state to $B_2$, who does the same and passes to $B_3$, and this process continues until the $n^{\text{th}}$ Bob. Later, Alice performs her measurement and compares her statistics with all the Bobs to ascertain the number of  Bobs that were able to detect shared entanglement with her. 

\subsection{Shared state and measurement strategy}
Consider the situation where Alice, $A$, and the first Bob, $B_1$, in the sequence of observers at the other end, share the  state,
\begin{equation}\label{initstate}
\begin{split}
   \rho_{AB_1}=\frac{1}{4}[\mathbb{I}_2 \otimes \mathbb{I}_2 -\cos \theta\sigma_1 \otimes \sigma_1 - \alpha \sin \theta \sigma_2 \otimes \sigma_2\\ - \alpha \sin \theta \sigma_3 \otimes \sigma_3], 
\end{split}
\end{equation}
where $\theta \in (0, \frac{\pi}{4}]$ and $\frac{1-\cos\theta}{2\sin\theta}<\alpha\leq 1$.
Note that $\rho_{AB_1}$ is entangled and admits an LHV model, i.e., \eqref{local} holds. For $\alpha=1$ and $\theta>\frac{\pi}{4}$, $\rho_{AB_1}$ becomes 
Bell inequality violating. 
Consider now the following measurement strategy adopted by $A$ and $B_k$, where $k=1,2,\dots,n$. $A$ applies measurement settings given by 
\begin{equation}
  \left\{\frac{\mathbb{I}_2+\sigma_i}{2}, ~\frac{\mathbb{I}_2-\sigma_i}{2}\right\},  
\end{equation}
 whereas $B_k$ applies measurement settings, $\{E^{(i)}_k,\mathbb{I}_2-E^{(i)}_k\}$, where 
\begin{equation}\label{Bmeas}
    E^{(i)}_k=\frac{\mathbb{I}_2+\lambda^{(i)}_k\sigma_i}{2},
\end{equation}
for $i=1,2,3$, and where $\lambda^{(i)}_k \in [0,1]$ is the sharpness parameter of the corresponding measurements.
Thus $A$ and $B_k$ can evaluate the expectation values of the entanglement witness operator given by
\begin{equation}
W_k=\frac{1}{4}\left[\mathbb{I}_2\otimes \mathbb{I}_2+ \sum_{i=1}^3\sigma_i \otimes \lambda^{(i)}_k\sigma_i\right],
\end{equation}
for any state shared between them using their measurement strategies. 

It is known that $\langle W_k \rangle \geq 0$ for all separable states for \(\lambda_k^{(i)} =1\) \cite{guhne02,guhne03}, while the cases for lower \(\lambda_k^{(i)}\) can be proven in the following way.

\begin{proof}
Every separable state can be expressed as a convex combination of pure product states. Therefore it is enough to prove that  $\langle W_k \rangle \geq 0$ for any pure product state. A two-qubit pure product state can be written in the form,
\begin{equation}
    \rho_{AB}^{sep} = \frac{\mathbb{I}_2+\Vec{r}_{A}.\Vec{\sigma}_{A}}{2} \otimes \frac{\mathbb{I}_2+\Vec{r}_{B}.\Vec{\sigma}_{B}}{2},
\end{equation}
where \(\Vec{\sigma}\)  is the vector of Pauli spin-1/2 matrices,
\(\Vec{r}_{A/B}\) are real three-dimensional unit vectors. 
Now,
\begin{equation}
    \text{Tr}\left(W_k \rho_{AB}^{sep}\right) = \frac{1}{4} \left(1 + \sum_{i=1}^3 \lambda^{(i)}_k r^{(i)}_A r^{(i)}_B\right).
\end{equation}
Let us define another vector \(\Vec{s}_B\) such that \(s_B^{(i)}=\lambda^{(i)}_k r_B^{(i)}\), \(i=1,2,3\). Note that $| \Vec{s}_{B}| \leq 1$, since $| \Vec{r}_B| = 1$ and $0\leq \lambda^{(i)}_k\leq 1$. Thus,
\begin{equation}
    \text{Tr}\left(W_k \rho_{AB}^{sep}\right) = \frac{1}{4} \left(1 + \Vec{r}_A.\Vec{s}_B\right)\geq 0.
\end{equation}
The inequality comes by applying the Cauchy-Schwarz inequality on 
$| \Vec{r}_{A}. \Vec{s}_{B}|$, and using $| \Vec{r}_{A}| = 1$ and $| \Vec{s}_{B}| \leq 1$.
\end{proof}


\noindent \emph{Post-measurement state:} After performing the measurement, $B_k$, passes his part of the shared state to the next Bob, i.e. $B_{k+1}.$ Let the state shared by $A$ and $B_k$ before \(B_k\) performs his measurement be  $\rho_{A{B_k}}$. Then the state shared by $A$ and $B_{k+1}$, before \(B_{k+1}\) performs his measurement, is given by the L\"uders rule: 
\begin{equation}
\begin{split}
    \rho_{AB_{k+1}}=\frac{1}{3}\sum_{i=1}^3\Big[\mathbb{I}_2\otimes\sqrt{E^{(i)}_k}.\rho_{AB_k}.\mathbb{I}_2\otimes\sqrt{E^{(i)}_k} \\ +~\mathbb{I}_2\otimes\sqrt{\mathbb{I}_2-E^{(i)}_k}.\rho_{AB_k}.\mathbb{I}_2\otimes\sqrt{\mathbb{I}_2-E^{(i)}_k}\Big].
\end{split}
\end{equation}
Note that the measurement settings used by any one Bob are equally probable, a consequence of the fact that each Bob acts independently, i.e., they do not know the measurement outcomes of the previous Bobs, or in other words, the measurement settings are ``unbiased''.

\section{Arbitrarily long sequence of entanglement detections for a ``local'' state}
In this section, we discuss how $A$ can detect entanglement, in an entangled state admitting an LHV model, with an arbitrary number of  sequential observers, $B_k$. To begin, $A$ and $B_1$ share the state $\rho_{AB_1}$,  and we assume that in the measurement strategy~\eqref{Bmeas} adopted by $B_k$, 
\begin{equation}
\label{mrignayani}
    \lambda^{(1)}_k=1 \text{ and }  \lambda^{(2)}_k=\lambda^{(3)}_k=\lambda_k,
\end{equation}
 where $\lambda_k \in (0,1)$. 
 While \(\lambda_k^{(i)}\), for \(i=1,2,3\), determine the sharpness of the measurement strategy of \(B_k\), in the case when~\eqref{mrignayani} holds, we consider \(\lambda_k\) as determining the sharpness.

Observers $A$ and $B_1$ can witness entanglement of $\rho_{AB_1}$, with $B_1$ using $\lambda_1$ as the sharpness parameter, if
$$\langle W_1 \rangle_{\rho_{AB_1}}=\frac{1}{4}\left[1-\cos \theta-2 \alpha \sin \theta \lambda_1\right]<0$$
$$ \implies \lambda_1 > \frac{1-\cos \theta}{2 \alpha \sin \theta}.$$
The state shared by $A$ and $B_k$ 
is given by
\begin{widetext}
\begin{equation}
\begin{split}
\rho_{AB_k}=\frac{1}{3}\Big[\frac{3+2\sqrt{1-\lambda^2_{k-1}}}{2}\rho_{AB_{k-1}} &+ \frac{1}{2}\mathbb{I}_2 \otimes \sigma_1 .\rho_{AB_{k-1}}.\mathbb{I}_2 \otimes \sigma_1 \\&+ \frac{1-\sqrt{1-\lambda^2_{k-1}}}{2} \left\{ \mathbb{I}_2 \otimes \sigma_2. \rho_{AB_{k-1}}.\mathbb{I}_2 \otimes \sigma_2+\mathbb{I}_2 \otimes \sigma_3 .\rho_{AB_{k-1}}.\mathbb{I}_2 \otimes \sigma_3 \right\} \Big].
\end{split}
\end{equation}
\end{widetext}
Therefore $A$ and $B_k$ can detect entanglement if $\langle W_k \rangle_{\rho_{AB_k}}<0$
$$\implies \lambda_k > \frac{1-\cos \theta\prod_{i=1}^{k-1}\frac{1+2\sqrt{1-\lambda^2_{i}}}{3}}{ 2 \alpha \sin \theta \prod_{i=1}^{k-1}\frac{1+\sqrt{1-\lambda^2_{i}}}{3}}.$$
For a given \(\epsilon > 0\),
 let us define the sequence, \(\{\lambda_k\}_{k=1}^{\infty}\), via the iterative rule, 
\begin{eqnarray}\label{lam}
\lambda_k &:=& (1+\epsilon)\frac{1-\cos \theta \prod_{i=1}^{k-1}\frac{1+2\sqrt{1-\lambda^2_{i}}}{3}}{ 2 \alpha \sin \theta \prod_{i=1}^{k-1}\frac{1+\sqrt{1-\lambda^2_{i}}}{3}}, \nonumber\\
\text{with }~\lambda_1 &:=&(1+\epsilon) \frac{1-\cos \theta}{2 \alpha \sin \theta},
\end{eqnarray}
if  $0<\lambda_{k-1}<1$ and undefined otherwise.
It is however easier to consider another sequence which bounds $\lambda_{k}$ from above. We define another iterative  sequence, \(\{\gamma_k\}_{k=1}^{\infty}\), via the iterative rule,
\begin{eqnarray}\label{gam}
\gamma_k &:=& (1+\epsilon)3^{k-1}\frac{1-(1- \frac{\theta^2}{2})\prod_{i=1}^{k-1}(1-\frac{2\gamma_i^2}{3})}{\alpha \theta \prod_{i=1}^{k-1} (2-\gamma_i^2)}, \nonumber\\ \text{with }~\gamma_1 &:=&(1+\epsilon)\frac{\theta}{2\alpha}, 
\end{eqnarray}
if  $0<\gamma_{k-1}<1$ and undefined otherwise.


\begin{lemma}\label{l1}
The sequence $\gamma_k(\theta)$, defined in Eq.~\eqref{gam}, is strictly increasing, in its range of validity. 
\end{lemma}
\begin{proof}
Note that $0<\gamma_1<1$ for $\theta \in (0,\frac{2\alpha}{1+\epsilon})$.
Let us assume that for some $k \in \mathbb{N}$ and for some range of $\theta$, $0<\gamma_k<1$. This also means for $j=1,2,\ldots k-1$, $0<\gamma_j<1$. Now   
$$\frac{\gamma_{k+1}}{\gamma_k}=\frac{3}{2-\gamma^2_{k}}\frac{1-(1- \frac{\theta^2}{2})\prod_{i=1}^{k}(1-\frac{2\gamma_i^2}{3})}{1-(1- \frac{\theta^2}{2})\prod_{i=1}^{k-1}(1-\frac{2\gamma_i^2}{3})}.$$
Since $0<(1-\frac{2\gamma_k^2}{3})<1$ for $0<\gamma_k<1$, therefore
$$\frac{\gamma_{k+1}}{\gamma_{k}}>\frac{3}{2}>1.$$
Thus, $\gamma_k$ is a strictly increasing sequence.
\end{proof}
In a similar fashion, it can  be demonstrated that the sequence $\lambda_k(\theta)$, defined in Eq.~\eqref{lam} is also a strictly increasing sequence.

\begin{lemma}\label{l2}
For $k \in \mathbb{N}$ and finite $\gamma_k(\theta)$, $\gamma_k(\theta) \geq \lambda_k(\theta)$.
\end{lemma}
\begin{proof}
For $\theta \in (0,\frac{\pi}{4}]$, $\cos\theta\geq 1- \frac{\theta^2}{2}$ and $\sin\theta \geq \frac{\theta}{2}$, and for $0<x <1$, $\sqrt{1-x^2}>{1-x^2}$. Using these inequalities, it can be shown that for each $k$, $$\gamma_k(\theta) \geq \lambda_k(\theta).$$
\end{proof}

\begin{theorem} 
For any $n \in \mathbb{N}$, there exists a local entangled state that is $n$-recyclable with respect to entanglement correlations. 
\end{theorem}

\begin{proof}
The theorem can be proved if it is possible to have $0<\lambda_n(\theta)<1$ for any $n\in \mathbb{N}$. However, it is enough to prove that $0<\gamma_n(\theta)<1$, since $0<\gamma_n(\theta)<1 \implies 0<\lambda_n(\theta)<1$ from Lemma~\ref{l2}, i.e., $\lambda_{n}$ lies in the valid region. 

We have 
$0<\gamma_1(\theta)<1$ for $\theta \in (0,\frac{2\alpha}{1+\epsilon})$. But the entangled state $\rho_{AB_1}$ admits LHV model if $\theta \in (0,\frac{\pi}{4}]$.
Therefore, let us denote the minimum of $\frac{2\alpha}{1+\epsilon}$ and  $\frac{\pi}{4}$ as $\theta_1$. 

Now 
$\gamma_1(\theta) \to 0$ as $\theta \to 0$. 
Also $0<\gamma_1(\theta)< 1$ for $\theta \in (0,\theta_1)$. For $\theta \in (0,\theta_1)$, $\gamma_2(\theta)$ is an odd function which goes to zero as $\theta \to 0$, and since $\gamma_2(\theta)>\gamma_1(\theta)$ (from Lemma \ref{l1}), there exists some $\theta_2 \in  (0,\theta_1)$ such that $0<\gamma_2(\theta)<1$ for $\theta \in (0,\theta_2)$.

Now assume that there exist some $\theta_k \in (0,\theta_{k-1})$ such that $0<\gamma_k(\theta)<1$, for $\theta \in (0,\theta_k)$. Therefore, $\gamma_{k+1}(\theta)$ is an odd polynomial of $\theta \in (0,\theta_k)$, going to zero as $\theta \to 0$. Again since, $\gamma_{k+1}(\theta)> \gamma_{k}(\theta),$ thus there exist a $\theta_{k+1} \in (0, \theta_k]$, such that $0<\gamma_{k+1}(\theta)<1$ for $\theta \in (0,\theta_{k+1}).$

\begin{figure}\label{fig}
    \centering
    \includegraphics[width=0.9\columnwidth]{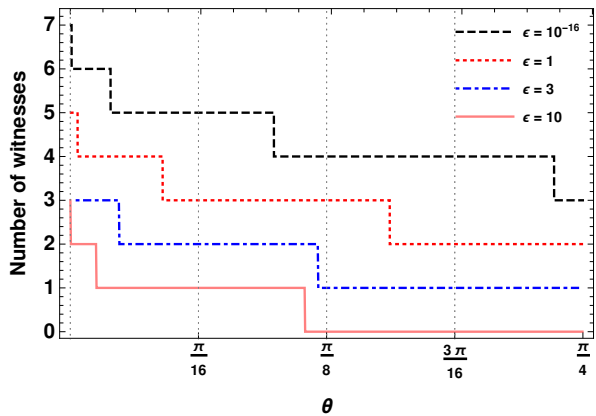}
    \caption{
    How many witnesses of entanglement? We plot here the number of Bobs who 
    can 
    witness entanglement with Alice, with respect to $\theta$, 
    for different values of $\epsilon$. 
There is a 
decrease in the number of Bobs with an increasing value of $\theta$. 
It is noteworthy that an increasing \(\theta\)  draws $\rho_{AB_{1}}$ towards the parameter regime in which it violates the CHSH Bell inequality, and away from the parameter point where it is separable.
All quantities used are dimensionless, except \(\theta\), which is measured in radians.
}
    \label{plot}
\end{figure}

Therefore, by induction, one can get to any number, $n \in \mathbb{N}$, of Bobs, such that $0<\lambda_1(\theta)<\lambda_2(\theta)<\ldots<\lambda_n(\theta)<1$ for $\theta \in (0,\theta_n)$.
\end{proof}
Notice that the number of Bobs, who can witness entanglement with a single Alice, \emph{increases} as $\theta \to 0$, which is the limit 
in which
entanglement of $\rho_{AB_1}$ also tends to zero. In fig. \ref{plot}, we  plot the number of Bobs who can witness entanglement with Alice, against $\theta$, for
different values $\epsilon$.  Clearly, as $\theta \to 0$ and $\epsilon \to 0$, the number of Bobs, who can detect entanglement via the witnessing procedure, increases.
It is important to point out that one can also reach arbitrary number of entanglement detections using the adopted measurement settings for Bell inequality-satisfying entangled states of the form, 
\begin{equation}
  \begin{split}
   \rho'_{AB_1}=\frac{1}{4}[\mathbb{I}_2 \otimes \mathbb{I}_2 -\cos \theta\sigma_1 \otimes \sigma_1 -  \alpha\sin \theta \sigma_2 \otimes \sigma_2\\ -~\beta \sin \theta \sigma_3 \otimes \sigma_3], 
\end{split}
\end{equation}
where $1\geq\alpha>\beta>0$, when $\theta \to 0$ and $\beta \to \alpha$.

\section{Conclusion}
We considered a sequential scenario, posed in \cite{Silva15} and established in \cite{Brown20}, that proves the possibility of arbitrarily many Bobs violating Bell inequalities independently by sharing a single entangled state with Alice.
We established that it is also possible to 
witness entanglement arbitrarily many times in the same scenario, even for shared states that does not violate a Bell inequality. 
We have constructed a state that meets these requirements and a suitable witness operator which utilizes a measurement strategy involving unsharp measurements for the Bobs. 
This shows that not only ``nonlocal'', i.e., Bell inequality violating quantum correlations, even ``local'' ones can be witnessed arbitrarily many times.
For the given example, we also provided an iterative bound on the sharpness of the measurements by any Bob, that depends on the initial state and the sharpness of the measurement performed by the previous Bobs. 
It was shown that using the given measurement strategy, the number of Bobs who can witness entanglement decreases as the state approaches its Bell inequality violating parameter regime, and away from the point on the parameter space where it is separable. 

We believe that the result that even ``local’’ entangled states - which are arguably weaker in quantum correlation content than the states that violate a Bell inequality - can allow an infinite number of sequential observers to detect entanglement, is a fundamentally interesting property about entanglement in general, and could lead to future applications.

\acknowledgements
We thank Peter Brown for an illuminating email. MP acknowledges the NCN (Poland) grant (grant number 2017/26/E/ST2/01008), and thanks Ray Ganardi for discussions. 
The research of CS was supported in part by the INFOSYS scholarship.
The authors from  Harish-Chandra Research Institute acknowledge partial support from the Department of Science and Technology, Government of India through the QuEST grant (grant number DST/ICPS/QUST/Theme-3/2019/120). 
\bibliography{apssamp}

\providecommand{\noopsort}[1]{}\providecommand{\singleletter}[1]{#1}%
\begin{thebibliography}{44}%
\makeatletter
\providecommand \@ifxundefined [1]{%
 \@ifx{#1\undefined}
}%
\providecommand \@ifnum [1]{%
 \ifnum #1\expandafter \@firstoftwo
 \else \expandafter \@secondoftwo
 \fi
}%
\providecommand \@ifx [1]{%
 \ifx #1\expandafter \@firstoftwo
 \else \expandafter \@secondoftwo
 \fi
}%
\providecommand \natexlab [1]{#1}%
\providecommand \enquote  [1]{``#1''}%
\providecommand \bibnamefont  [1]{#1}%
\providecommand \bibfnamefont [1]{#1}%
\providecommand \citenamefont [1]{#1}%
\providecommand \href@noop [0]{\@secondoftwo}%
\providecommand \href [0]{\begingroup \@sanitize@url \@href}%
\providecommand \@href[1]{\@@startlink{#1}\@@href}%
\providecommand \@@href[1]{\endgroup#1\@@endlink}%
\providecommand \@sanitize@url [0]{\catcode `\\12\catcode `\$12\catcode
  `\&12\catcode `\#12\catcode `\^12\catcode `\_12\catcode `\%12\relax}%
\providecommand \@@startlink[1]{}%
\providecommand \@@endlink[0]{}%
\providecommand \url  [0]{\begingroup\@sanitize@url \@url }%
\providecommand \@url [1]{\endgroup\@href {#1}{\urlprefix }}%
\providecommand \urlprefix  [0]{URL }%
\providecommand \Eprint [0]{\href }%
\providecommand \doibase [0]{http://dx.doi.org/}%
\providecommand \selectlanguage [0]{\@gobble}%
\providecommand \bibinfo  [0]{\@secondoftwo}%
\providecommand \bibfield  [0]{\@secondoftwo}%
\providecommand \translation [1]{[#1]}%
\providecommand \BibitemOpen [0]{}%
\providecommand \bibitemStop [0]{}%
\providecommand \bibitemNoStop [0]{.\EOS\space}%
\providecommand \EOS [0]{\spacefactor3000\relax}%
\providecommand \BibitemShut  [1]{\csname bibitem#1\endcsname}%
\let\auto@bib@innerbib\@empty
\bibitem [{\citenamefont {Horodecki}\ \emph {et~al.}(2009)\citenamefont
  {Horodecki}, \citenamefont {Horodecki}, \citenamefont {Horodecki},\ and\
  \citenamefont {Horodecki}}]{horodecki09}%
  \BibitemOpen
  \bibfield  {author} {\bibinfo {author} {\bibfnamefont {Ryszard}\ \bibnamefont
  {Horodecki}}, \bibinfo {author} {\bibfnamefont {Pawe\l{}}\ \bibnamefont
  {Horodecki}}, \bibinfo {author} {\bibfnamefont {Micha\l{}}\ \bibnamefont
  {Horodecki}}, \ and\ \bibinfo {author} {\bibfnamefont {Karol}\ \bibnamefont
  {Horodecki}},\ }\bibfield  {title} {\enquote {\bibinfo {title} {Quantum
  entanglement},}\ }\href {\doibase 10.1103/RevModPhys.81.865} {\bibfield
  {journal} {\bibinfo  {journal} {Rev. Mod. Phys.}\ }\textbf {\bibinfo {volume}
  {81}},\ \bibinfo {pages} {865--942} (\bibinfo {year} {2009})}\BibitemShut
  {NoStop}%
\bibitem [{\citenamefont {Gühne}\ and\ \citenamefont {Tóth}(2009)}]{Guhne09}%
  \BibitemOpen
  \bibfield  {author} {\bibinfo {author} {\bibfnamefont {Otfried}\ \bibnamefont
  {Gühne}}\ and\ \bibinfo {author} {\bibfnamefont {Géza}\ \bibnamefont
  {Tóth}},\ }\bibfield  {title} {\enquote {\bibinfo {title} {Entanglement
  detection},}\ }\href {\doibase https://doi.org/10.1016/j.physrep.2009.02.004}
  {\bibfield  {journal} {\bibinfo  {journal} {Physics Reports}\ }\textbf
  {\bibinfo {volume} {474}},\ \bibinfo {pages} {1--75} (\bibinfo {year}
  {2009})}\BibitemShut {NoStop}%
\bibitem [{\citenamefont {Bell}(1964)}]{Bell64}%
  \BibitemOpen
  \bibfield  {author} {\bibinfo {author} {\bibfnamefont {J.~S.}\ \bibnamefont
  {Bell}},\ }\bibfield  {title} {\enquote {\bibinfo {title} {On the {E}instein
  {P}odolsky {R}osen paradox},}\ }\href {\doibase
  10.1103/PhysicsPhysiqueFizika.1.195} {\bibfield  {journal} {\bibinfo
  {journal} {Physics Physique Fizika}\ }\textbf {\bibinfo {volume} {1}},\
  \bibinfo {pages} {195--200} (\bibinfo {year} {1964})}\BibitemShut {NoStop}%
\bibitem [{\citenamefont {Brunner}\ \emph {et~al.}(2014)\citenamefont
  {Brunner}, \citenamefont {Cavalcanti}, \citenamefont {Pironio}, \citenamefont
  {Scarani},\ and\ \citenamefont {Wehner}}]{Brunner14}%
  \BibitemOpen
  \bibfield  {author} {\bibinfo {author} {\bibfnamefont {Nicolas}\ \bibnamefont
  {Brunner}}, \bibinfo {author} {\bibfnamefont {Daniel}\ \bibnamefont
  {Cavalcanti}}, \bibinfo {author} {\bibfnamefont {Stefano}\ \bibnamefont
  {Pironio}}, \bibinfo {author} {\bibfnamefont {Valerio}\ \bibnamefont
  {Scarani}}, \ and\ \bibinfo {author} {\bibfnamefont {Stephanie}\ \bibnamefont
  {Wehner}},\ }\bibfield  {title} {\enquote {\bibinfo {title} {Bell
  nonlocality},}\ }\href {\doibase 10.1103/RevModPhys.86.419} {\bibfield
  {journal} {\bibinfo  {journal} {Rev. Mod. Phys.}\ }\textbf {\bibinfo {volume}
  {86}},\ \bibinfo {pages} {419--478} (\bibinfo {year} {2014})}\BibitemShut
  {NoStop}%
\bibitem [{\citenamefont {Werner}(1989)}]{Werner89}%
  \BibitemOpen
  \bibfield  {author} {\bibinfo {author} {\bibfnamefont {Reinhard~F.}\
  \bibnamefont {Werner}},\ }\bibfield  {title} {\enquote {\bibinfo {title}
  {Quantum states with {E}instein-{P}odolsky-{R}osen correlations admitting a
  hidden-variable model},}\ }\href {\doibase 10.1103/PhysRevA.40.4277}
  {\bibfield  {journal} {\bibinfo  {journal} {Phys. Rev. A}\ }\textbf {\bibinfo
  {volume} {40}},\ \bibinfo {pages} {4277--4281} (\bibinfo {year}
  {1989})}\BibitemShut {NoStop}%
\bibitem [{\citenamefont {Barrett}(2002)}]{Barrett02}%
  \BibitemOpen
  \bibfield  {author} {\bibinfo {author} {\bibfnamefont {Jonathan}\
  \bibnamefont {Barrett}},\ }\bibfield  {title} {\enquote {\bibinfo {title}
  {Nonsequential positive-operator-valued measurements on entangled mixed
  states do not always violate a {B}ell inequality},}\ }\href {\doibase
  10.1103/PhysRevA.65.042302} {\bibfield  {journal} {\bibinfo  {journal} {Phys.
  Rev. A}\ }\textbf {\bibinfo {volume} {65}},\ \bibinfo {pages} {042302}
  (\bibinfo {year} {2002})}\BibitemShut {NoStop}%
\bibitem [{\citenamefont {Mayers}\ and\ \citenamefont {Yao}(1998)}]{Mayers98}%
  \BibitemOpen
  \bibfield  {author} {\bibinfo {author} {\bibfnamefont {D.}~\bibnamefont
  {Mayers}}\ and\ \bibinfo {author} {\bibfnamefont {A.}~\bibnamefont {Yao}},\
  }\bibfield  {title} {\enquote {\bibinfo {title} {Quantum cryptography with
  imperfect apparatus},}\ }\href@noop {} {\bibfield  {journal} {\bibinfo
  {journal} {Proceedings 39th Annual Symposium on Foundations of Computer
  Science (Cat. No.98CB36280)}\ ,\ \bibinfo {pages} {503--509}} (\bibinfo
  {year} {1998})}\BibitemShut {NoStop}%
\bibitem [{\citenamefont {Barrett}\ \emph {et~al.}(2005)\citenamefont
  {Barrett}, \citenamefont {Hardy},\ and\ \citenamefont {Kent}}]{Barrett05}%
  \BibitemOpen
  \bibfield  {author} {\bibinfo {author} {\bibfnamefont {Jonathan}\
  \bibnamefont {Barrett}}, \bibinfo {author} {\bibfnamefont {Lucien}\
  \bibnamefont {Hardy}}, \ and\ \bibinfo {author} {\bibfnamefont {Adrian}\
  \bibnamefont {Kent}},\ }\bibfield  {title} {\enquote {\bibinfo {title} {No
  signaling and quantum key distribution},}\ }\href {\doibase
  10.1103/PhysRevLett.95.010503} {\bibfield  {journal} {\bibinfo  {journal}
  {Phys. Rev. Lett.}\ }\textbf {\bibinfo {volume} {95}},\ \bibinfo {pages}
  {010503} (\bibinfo {year} {2005})}\BibitemShut {NoStop}%
\bibitem [{\citenamefont {Ac\'{\i}n}\ \emph {et~al.}(2007)\citenamefont
  {Ac\'{\i}n}, \citenamefont {Brunner}, \citenamefont {Gisin}, \citenamefont
  {Massar}, \citenamefont {Pironio},\ and\ \citenamefont {Scarani}}]{Acin07}%
  \BibitemOpen
  \bibfield  {author} {\bibinfo {author} {\bibfnamefont {Antonio}\ \bibnamefont
  {Ac\'{\i}n}}, \bibinfo {author} {\bibfnamefont {Nicolas}\ \bibnamefont
  {Brunner}}, \bibinfo {author} {\bibfnamefont {Nicolas}\ \bibnamefont
  {Gisin}}, \bibinfo {author} {\bibfnamefont {Serge}\ \bibnamefont {Massar}},
  \bibinfo {author} {\bibfnamefont {Stefano}\ \bibnamefont {Pironio}}, \ and\
  \bibinfo {author} {\bibfnamefont {Valerio}\ \bibnamefont {Scarani}},\
  }\bibfield  {title} {\enquote {\bibinfo {title} {Device-independent security
  of quantum cryptography against collective attacks},}\ }\href {\doibase
  10.1103/PhysRevLett.98.230501} {\bibfield  {journal} {\bibinfo  {journal}
  {Phys. Rev. Lett.}\ }\textbf {\bibinfo {volume} {98}},\ \bibinfo {pages}
  {230501} (\bibinfo {year} {2007})}\BibitemShut {NoStop}%
\bibitem [{\citenamefont {Colbeck}\ and\ \citenamefont
  {Renner}(2012)}]{Colbeck12}%
  \BibitemOpen
  \bibfield  {author} {\bibinfo {author} {\bibfnamefont {Roger}\ \bibnamefont
  {Colbeck}}\ and\ \bibinfo {author} {\bibfnamefont {Renato}\ \bibnamefont
  {Renner}},\ }\bibfield  {title} {\enquote {\bibinfo {title} {Free randomness
  can be amplified},}\ }\href {\doibase 10.1038/nphys2300} {\bibfield
  {journal} {\bibinfo  {journal} {Nature Physics}\ }\textbf {\bibinfo {volume}
  {8}},\ \bibinfo {pages} {450–453} (\bibinfo {year} {2012})}\BibitemShut
  {NoStop}%
\bibitem [{\citenamefont {Colbeck}(2011)}]{Colbeck111}%
  \BibitemOpen
  \bibfield  {author} {\bibinfo {author} {\bibfnamefont {Roger}\ \bibnamefont
  {Colbeck}},\ }\href@noop {} {\enquote {\bibinfo {title} {Quantum and
  relativistic protocols for secure multi-party computation},}\ } (\bibinfo
  {year} {2011}),\ \Eprint {http://arxiv.org/abs/0911.3814} {arXiv:0911.3814
  [quant-ph]} \BibitemShut {NoStop}%
\bibitem [{\citenamefont {Colbeck}\ and\ \citenamefont
  {Kent}(2011)}]{Colbeck112}%
  \BibitemOpen
  \bibfield  {author} {\bibinfo {author} {\bibfnamefont {Roger}\ \bibnamefont
  {Colbeck}}\ and\ \bibinfo {author} {\bibfnamefont {Adrian}\ \bibnamefont
  {Kent}},\ }\bibfield  {title} {\enquote {\bibinfo {title} {Private randomness
  expansion with untrusted devices},}\ }\href {\doibase
  10.1088/1751-8113/44/9/095305} {\bibfield  {journal} {\bibinfo  {journal}
  {Journal of Physics A: Mathematical and Theoretical}\ }\textbf {\bibinfo
  {volume} {44}},\ \bibinfo {pages} {095305} (\bibinfo {year}
  {2011})}\BibitemShut {NoStop}%
\bibitem [{\citenamefont {Pironio}\ \emph {et~al.}(2010)\citenamefont
  {Pironio}, \citenamefont {Ac{\'i}n}, \citenamefont {Massar}, \citenamefont
  {de~la Giroday}, \citenamefont {Matsukevich}, \citenamefont {Maunz},
  \citenamefont {Olmschenk}, \citenamefont {Hayes}, \citenamefont {Luo},
  \citenamefont {Manning},\ and\ \citenamefont {Monroe}}]{Pironio10}%
  \BibitemOpen
  \bibfield  {author} {\bibinfo {author} {\bibfnamefont {S.}~\bibnamefont
  {Pironio}}, \bibinfo {author} {\bibfnamefont {A.}~\bibnamefont {Ac{\'i}n}},
  \bibinfo {author} {\bibfnamefont {S.}~\bibnamefont {Massar}}, \bibinfo
  {author} {\bibfnamefont {A.~Boyer}\ \bibnamefont {de~la Giroday}}, \bibinfo
  {author} {\bibfnamefont {D.~N.}\ \bibnamefont {Matsukevich}}, \bibinfo
  {author} {\bibfnamefont {P.}~\bibnamefont {Maunz}}, \bibinfo {author}
  {\bibfnamefont {S.}~\bibnamefont {Olmschenk}}, \bibinfo {author}
  {\bibfnamefont {D.}~\bibnamefont {Hayes}}, \bibinfo {author} {\bibfnamefont
  {L.}~\bibnamefont {Luo}}, \bibinfo {author} {\bibfnamefont {T.~A.}\
  \bibnamefont {Manning}}, \ and\ \bibinfo {author} {\bibfnamefont
  {C.}~\bibnamefont {Monroe}},\ }\bibfield  {title} {\enquote {\bibinfo {title}
  {Random numbers certified by bell's theorem},}\ }\href {\doibase
  10.1038/nature09008} {\bibfield  {journal} {\bibinfo  {journal} {Nature}\
  }\textbf {\bibinfo {volume} {464}},\ \bibinfo {pages} {1021--1024} (\bibinfo
  {year} {2010})}\BibitemShut {NoStop}%
\bibitem [{\citenamefont {Bennett}\ \emph {et~al.}(1993)\citenamefont
  {Bennett}, \citenamefont {Brassard}, \citenamefont {Cr\'epeau}, \citenamefont
  {Jozsa}, \citenamefont {Peres},\ and\ \citenamefont {Wootters}}]{Ben93}%
  \BibitemOpen
  \bibfield  {author} {\bibinfo {author} {\bibfnamefont {Charles~H.}\
  \bibnamefont {Bennett}}, \bibinfo {author} {\bibfnamefont {Gilles}\
  \bibnamefont {Brassard}}, \bibinfo {author} {\bibfnamefont {Claude}\
  \bibnamefont {Cr\'epeau}}, \bibinfo {author} {\bibfnamefont {Richard}\
  \bibnamefont {Jozsa}}, \bibinfo {author} {\bibfnamefont {Asher}\ \bibnamefont
  {Peres}}, \ and\ \bibinfo {author} {\bibfnamefont {William~K.}\ \bibnamefont
  {Wootters}},\ }\bibfield  {title} {\enquote {\bibinfo {title} {Teleporting an
  unknown quantum state via dual classical and {E}instein-{P}odolsky-{R}osen
  channels},}\ }\href {\doibase 10.1103/PhysRevLett.70.1895} {\bibfield
  {journal} {\bibinfo  {journal} {Phys. Rev. Lett.}\ }\textbf {\bibinfo
  {volume} {70}},\ \bibinfo {pages} {1895--1899} (\bibinfo {year}
  {1993})}\BibitemShut {NoStop}%
\bibitem [{\citenamefont {Horodecki}\ \emph {et~al.}(1999)\citenamefont
  {Horodecki}, \citenamefont {Horodecki},\ and\ \citenamefont
  {Horodecki}}]{Hor99}%
  \BibitemOpen
  \bibfield  {author} {\bibinfo {author} {\bibfnamefont {Micha\l{}}\
  \bibnamefont {Horodecki}}, \bibinfo {author} {\bibfnamefont {Pawe\l{}}\
  \bibnamefont {Horodecki}}, \ and\ \bibinfo {author} {\bibfnamefont {Ryszard}\
  \bibnamefont {Horodecki}},\ }\bibfield  {title} {\enquote {\bibinfo {title}
  {General teleportation channel, singlet fraction, and quasidistillation},}\
  }\href {\doibase 10.1103/PhysRevA.60.1888} {\bibfield  {journal} {\bibinfo
  {journal} {Phys. Rev. A}\ }\textbf {\bibinfo {volume} {60}},\ \bibinfo
  {pages} {1888--1898} (\bibinfo {year} {1999})}\BibitemShut {NoStop}%
\bibitem [{\citenamefont {Nielsen}(2002)}]{Nielsen02}%
  \BibitemOpen
  \bibfield  {author} {\bibinfo {author} {\bibfnamefont {Michael~A}\
  \bibnamefont {Nielsen}},\ }\bibfield  {title} {\enquote {\bibinfo {title} {A
  simple formula for the average gate fidelity of a quantum dynamical
  operation},}\ }\href {\doibase https://doi.org/10.1016/S0375-9601(02)01272-0}
  {\bibfield  {journal} {\bibinfo  {journal} {Physics Letters A}\ }\textbf
  {\bibinfo {volume} {303}},\ \bibinfo {pages} {249--252} (\bibinfo {year}
  {2002})}\BibitemShut {NoStop}%
\bibitem [{\citenamefont {Silva}\ \emph {et~al.}(2015)\citenamefont {Silva},
  \citenamefont {Gisin}, \citenamefont {Guryanova},\ and\ \citenamefont
  {Popescu}}]{Silva15}%
  \BibitemOpen
  \bibfield  {author} {\bibinfo {author} {\bibfnamefont {Ralph}\ \bibnamefont
  {Silva}}, \bibinfo {author} {\bibfnamefont {Nicolas}\ \bibnamefont {Gisin}},
  \bibinfo {author} {\bibfnamefont {Yelena}\ \bibnamefont {Guryanova}}, \ and\
  \bibinfo {author} {\bibfnamefont {Sandu}\ \bibnamefont {Popescu}},\
  }\bibfield  {title} {\enquote {\bibinfo {title} {Multiple observers can share
  the nonlocality of half of an entangled pair by using optimal weak
  measurements},}\ }\href {\doibase 10.1103/PhysRevLett.114.250401} {\bibfield
  {journal} {\bibinfo  {journal} {Phys. Rev. Lett.}\ }\textbf {\bibinfo
  {volume} {114}},\ \bibinfo {pages} {250401} (\bibinfo {year}
  {2015})}\BibitemShut {NoStop}%
\bibitem [{\citenamefont {Clauser}\ \emph {et~al.}(1969)\citenamefont
  {Clauser}, \citenamefont {Horne}, \citenamefont {Shimony},\ and\
  \citenamefont {Holt}}]{Clauser69}%
  \BibitemOpen
  \bibfield  {author} {\bibinfo {author} {\bibfnamefont {John~F.}\ \bibnamefont
  {Clauser}}, \bibinfo {author} {\bibfnamefont {Michael~A.}\ \bibnamefont
  {Horne}}, \bibinfo {author} {\bibfnamefont {Abner}\ \bibnamefont {Shimony}},
  \ and\ \bibinfo {author} {\bibfnamefont {Richard~A.}\ \bibnamefont {Holt}},\
  }\bibfield  {title} {\enquote {\bibinfo {title} {Proposed experiment to test
  local hidden-variable theories},}\ }\href {\doibase
  10.1103/PhysRevLett.23.880} {\bibfield  {journal} {\bibinfo  {journal} {Phys.
  Rev. Lett.}\ }\textbf {\bibinfo {volume} {23}},\ \bibinfo {pages} {880--884}
  (\bibinfo {year} {1969})}\BibitemShut {NoStop}%
\bibitem [{\citenamefont {Brown}\ and\ \citenamefont
  {Colbeck}(2020)}]{Brown20}%
  \BibitemOpen
  \bibfield  {author} {\bibinfo {author} {\bibfnamefont {Peter~J.}\
  \bibnamefont {Brown}}\ and\ \bibinfo {author} {\bibfnamefont {Roger}\
  \bibnamefont {Colbeck}},\ }\bibfield  {title} {\enquote {\bibinfo {title}
  {Arbitrarily many independent observers can share the nonlocality of a single
  maximally entangled qubit pair},}\ }\href {\doibase
  10.1103/PhysRevLett.125.090401} {\bibfield  {journal} {\bibinfo  {journal}
  {Phys. Rev. Lett.}\ }\textbf {\bibinfo {volume} {125}},\ \bibinfo {pages}
  {090401} (\bibinfo {year} {2020})}\BibitemShut {NoStop}%
\bibitem [{\citenamefont {Mal}\ \emph {et~al.}(2016)\citenamefont {Mal},
  \citenamefont {Majumdar},\ and\ \citenamefont {Home}}]{Mal16}%
  \BibitemOpen
  \bibfield  {author} {\bibinfo {author} {\bibfnamefont {Shiladitya}\
  \bibnamefont {Mal}}, \bibinfo {author} {\bibfnamefont {Archan~S.}\
  \bibnamefont {Majumdar}}, \ and\ \bibinfo {author} {\bibfnamefont {Dipankar}\
  \bibnamefont {Home}},\ }\bibfield  {title} {\enquote {\bibinfo {title}
  {Sharing of nonlocality of a single member of an entangled pair of qubits is
  not possible by more than two unbiased observers on the other wing},}\ }\href
  {\doibase 10.3390/math4030048} {\bibfield  {journal} {\bibinfo  {journal}
  {Mathematics}\ }\textbf {\bibinfo {volume} {4}} (\bibinfo {year} {2016}),\
  10.3390/math4030048}\BibitemShut {NoStop}%
\bibitem [{\citenamefont {Hu}\ \emph {et~al.}(2018)\citenamefont {Hu},
  \citenamefont {Zhou}, \citenamefont {Hu}, \citenamefont {Li}, \citenamefont
  {Guo},\ and\ \citenamefont {Zhang}}]{Hu18}%
  \BibitemOpen
  \bibfield  {author} {\bibinfo {author} {\bibfnamefont {Meng-Jun}\
  \bibnamefont {Hu}}, \bibinfo {author} {\bibfnamefont {Zhi-Yuan}\ \bibnamefont
  {Zhou}}, \bibinfo {author} {\bibfnamefont {Xiao-Min}\ \bibnamefont {Hu}},
  \bibinfo {author} {\bibfnamefont {Chuan-Feng}\ \bibnamefont {Li}}, \bibinfo
  {author} {\bibfnamefont {Guang-Can}\ \bibnamefont {Guo}}, \ and\ \bibinfo
  {author} {\bibfnamefont {Yong-Sheng}\ \bibnamefont {Zhang}},\ }\bibfield
  {title} {\enquote {\bibinfo {title} {Observation of non-locality sharing
  among three observers with one entangled pair via optimal weak
  measurement},}\ }\href {\doibase 10.1038/s41534-018-0115-x} {\bibfield
  {journal} {\bibinfo  {journal} {npj Quantum Information}\ }\textbf {\bibinfo
  {volume} {4}},\ \bibinfo {pages} {63} (\bibinfo {year} {2018})}\BibitemShut
  {NoStop}%
\bibitem [{\citenamefont {Saha}\ \emph {et~al.}(2019)\citenamefont {Saha},
  \citenamefont {Das}, \citenamefont {Sasmal}, \citenamefont {Sarkar},
  \citenamefont {Mukherjee}, \citenamefont {Roy},\ and\ \citenamefont
  {Bhattacharya}}]{Saha19}%
  \BibitemOpen
  \bibfield  {author} {\bibinfo {author} {\bibfnamefont {Sutapa}\ \bibnamefont
  {Saha}}, \bibinfo {author} {\bibfnamefont {Debarshi}\ \bibnamefont {Das}},
  \bibinfo {author} {\bibfnamefont {Souradeep}\ \bibnamefont {Sasmal}},
  \bibinfo {author} {\bibfnamefont {Debasis}\ \bibnamefont {Sarkar}}, \bibinfo
  {author} {\bibfnamefont {Kaushiki}\ \bibnamefont {Mukherjee}}, \bibinfo
  {author} {\bibfnamefont {Arup}\ \bibnamefont {Roy}}, \ and\ \bibinfo {author}
  {\bibfnamefont {Some~Sankar}\ \bibnamefont {Bhattacharya}},\ }\bibfield
  {title} {\enquote {\bibinfo {title} {Sharing of tripartite nonlocality by
  multiple observers measuring sequentially at one side},}\ }\href {\doibase
  10.1007/s11128-018-2161-x} {\bibfield  {journal} {\bibinfo  {journal}
  {Quantum Information Processing}\ }\textbf {\bibinfo {volume} {18}},\
  \bibinfo {pages} {42} (\bibinfo {year} {2019})}\BibitemShut {NoStop}%
\bibitem [{\citenamefont {Das}\ \emph {et~al.}(2019)\citenamefont {Das},
  \citenamefont {Ghosal}, \citenamefont {Sasmal}, \citenamefont {Mal},\ and\
  \citenamefont {Majumdar}}]{Das19}%
  \BibitemOpen
  \bibfield  {author} {\bibinfo {author} {\bibfnamefont {Debarshi}\
  \bibnamefont {Das}}, \bibinfo {author} {\bibfnamefont {Arkaprabha}\
  \bibnamefont {Ghosal}}, \bibinfo {author} {\bibfnamefont {Souradeep}\
  \bibnamefont {Sasmal}}, \bibinfo {author} {\bibfnamefont {Shiladitya}\
  \bibnamefont {Mal}}, \ and\ \bibinfo {author} {\bibfnamefont {A.~S.}\
  \bibnamefont {Majumdar}},\ }\bibfield  {title} {\enquote {\bibinfo {title}
  {Facets of bipartite nonlocality sharing by multiple observers via sequential
  measurements},}\ }\href {\doibase 10.1103/PhysRevA.99.022305} {\bibfield
  {journal} {\bibinfo  {journal} {Phys. Rev. A}\ }\textbf {\bibinfo {volume}
  {99}},\ \bibinfo {pages} {022305} (\bibinfo {year} {2019})}\BibitemShut
  {NoStop}%
\bibitem [{\citenamefont {Foletto}\ \emph {et~al.}(2020)\citenamefont
  {Foletto}, \citenamefont {Calderaro}, \citenamefont {Tavakoli}, \citenamefont
  {Schiavon}, \citenamefont {Picciariello}, \citenamefont {Cabello},
  \citenamefont {Villoresi},\ and\ \citenamefont {Vallone}}]{Vallone20}%
  \BibitemOpen
  \bibfield  {author} {\bibinfo {author} {\bibfnamefont {Giulio}\ \bibnamefont
  {Foletto}}, \bibinfo {author} {\bibfnamefont {Luca}\ \bibnamefont
  {Calderaro}}, \bibinfo {author} {\bibfnamefont {Armin}\ \bibnamefont
  {Tavakoli}}, \bibinfo {author} {\bibfnamefont {Matteo}\ \bibnamefont
  {Schiavon}}, \bibinfo {author} {\bibfnamefont {Francesco}\ \bibnamefont
  {Picciariello}}, \bibinfo {author} {\bibfnamefont {Ad\'an}\ \bibnamefont
  {Cabello}}, \bibinfo {author} {\bibfnamefont {Paolo}\ \bibnamefont
  {Villoresi}}, \ and\ \bibinfo {author} {\bibfnamefont {Giuseppe}\
  \bibnamefont {Vallone}},\ }\bibfield  {title} {\enquote {\bibinfo {title}
  {Experimental certification of sustained entanglement and nonlocality after
  sequential measurements},}\ }\href {\doibase
  10.1103/PhysRevApplied.13.044008} {\bibfield  {journal} {\bibinfo  {journal}
  {Phys. Rev. Applied}\ }\textbf {\bibinfo {volume} {13}},\ \bibinfo {pages}
  {044008} (\bibinfo {year} {2020})}\BibitemShut {NoStop}%
\bibitem [{\citenamefont {Zhang}\ and\ \citenamefont {Fei}(2021)}]{Fei21}%
  \BibitemOpen
  \bibfield  {author} {\bibinfo {author} {\bibfnamefont {Tinggui}\ \bibnamefont
  {Zhang}}\ and\ \bibinfo {author} {\bibfnamefont {Shao-Ming}\ \bibnamefont
  {Fei}},\ }\bibfield  {title} {\enquote {\bibinfo {title} {Sharing quantum
  nonlocality and genuine nonlocality with independent observables},}\ }\href
  {\doibase 10.1103/PhysRevA.103.032216} {\bibfield  {journal} {\bibinfo
  {journal} {Phys. Rev. A}\ }\textbf {\bibinfo {volume} {103}},\ \bibinfo
  {pages} {032216} (\bibinfo {year} {2021})}\BibitemShut {NoStop}%
\bibitem [{\citenamefont {Roy}\ \emph {et~al.}(2020)\citenamefont {Roy},
  \citenamefont {Kumari}, \citenamefont {Mal},\ and\ \citenamefont
  {De}}]{roy20}%
  \BibitemOpen
  \bibfield  {author} {\bibinfo {author} {\bibfnamefont {Saptarshi}\
  \bibnamefont {Roy}}, \bibinfo {author} {\bibfnamefont {Asmita}\ \bibnamefont
  {Kumari}}, \bibinfo {author} {\bibfnamefont {Shiladitya}\ \bibnamefont
  {Mal}}, \ and\ \bibinfo {author} {\bibfnamefont {Aditi~Sen}\ \bibnamefont
  {De}},\ }\href@noop {} {\enquote {\bibinfo {title} {Robustness of higher
  dimensional nonlocality against dual noise and sequential measurements},}\ }
  (\bibinfo {year} {2020}),\ \Eprint {http://arxiv.org/abs/2012.12200}
  {arXiv:2012.12200 [quant-ph]} \BibitemShut {NoStop}%
\bibitem [{\citenamefont {Cabello}(2021)}]{cabello20}%
  \BibitemOpen
  \bibfield  {author} {\bibinfo {author} {\bibfnamefont {Adán}\ \bibnamefont
  {Cabello}},\ }\href@noop {} {\enquote {\bibinfo {title} {Bell nonlocality
  between sequential pairs of observers},}\ } (\bibinfo {year} {2021}),\
  \Eprint {http://arxiv.org/abs/2103.11844} {arXiv:2103.11844 [quant-ph]}
  \BibitemShut {NoStop}%
\bibitem [{\citenamefont {Ren}\ \emph {et~al.}(2021)\citenamefont {Ren},
  \citenamefont {Liu}, \citenamefont {Hou}, \citenamefont {Feng},\ and\
  \citenamefont {Zhou}}]{ren21}%
  \BibitemOpen
  \bibfield  {author} {\bibinfo {author} {\bibfnamefont {Changliang}\
  \bibnamefont {Ren}}, \bibinfo {author} {\bibfnamefont {Xiaowei}\ \bibnamefont
  {Liu}}, \bibinfo {author} {\bibfnamefont {Wenlin}\ \bibnamefont {Hou}},
  \bibinfo {author} {\bibfnamefont {Tianfeng}\ \bibnamefont {Feng}}, \ and\
  \bibinfo {author} {\bibfnamefont {Xiaoqi}\ \bibnamefont {Zhou}},\ }\href@noop
  {} {\enquote {\bibinfo {title} {Non-locality sharing for a three-qubit system
  via multilateral sequential measurements},}\ } (\bibinfo {year} {2021}),\
  \Eprint {http://arxiv.org/abs/2105.03709} {arXiv:2105.03709 [quant-ph]}
  \BibitemShut {NoStop}%
\bibitem [{\citenamefont {Zhu}\ \emph {et~al.}(2021)\citenamefont {Zhu},
  \citenamefont {Hu}, \citenamefont {Guo}, \citenamefont {Li},\ and\
  \citenamefont {Zhang}}]{zhu21}%
  \BibitemOpen
  \bibfield  {author} {\bibinfo {author} {\bibfnamefont {Jie}\ \bibnamefont
  {Zhu}}, \bibinfo {author} {\bibfnamefont {Meng-Jun}\ \bibnamefont {Hu}},
  \bibinfo {author} {\bibfnamefont {Guang-Can}\ \bibnamefont {Guo}}, \bibinfo
  {author} {\bibfnamefont {Chuan-Feng}\ \bibnamefont {Li}}, \ and\ \bibinfo
  {author} {\bibfnamefont {Yong-Sheng}\ \bibnamefont {Zhang}},\ }\href@noop {}
  {\enquote {\bibinfo {title} {Einstein-podolsky-rosen steering in two-sided
  sequential measurements with one entangled pair},}\ } (\bibinfo {year}
  {2021}),\ \Eprint {http://arxiv.org/abs/2102.02550} {arXiv:2102.02550
  [quant-ph]} \BibitemShut {NoStop}%
\bibitem [{\citenamefont {Cheng}\ \emph
  {et~al.}(2021{\natexlab{a}})\citenamefont {Cheng}, \citenamefont {Liu},
  \citenamefont {Baker},\ and\ \citenamefont {Hall}}]{Cheng21}%
  \BibitemOpen
  \bibfield  {author} {\bibinfo {author} {\bibfnamefont {Shuming}\ \bibnamefont
  {Cheng}}, \bibinfo {author} {\bibfnamefont {Lijun}\ \bibnamefont {Liu}},
  \bibinfo {author} {\bibfnamefont {Travis~J.}\ \bibnamefont {Baker}}, \ and\
  \bibinfo {author} {\bibfnamefont {Michael J.~W.}\ \bibnamefont {Hall}},\
  }\href@noop {} {\enquote {\bibinfo {title} {Recycling qubits for the
  generation of {B}ell nonlocality between independent sequential observers},}\
  } (\bibinfo {year} {2021}{\natexlab{a}}),\ \Eprint
  {http://arxiv.org/abs/2109.03472} {arXiv:2109.03472 [quant-ph]} \BibitemShut
  {NoStop}%
\bibitem [{\citenamefont {Cheng}\ \emph
  {et~al.}(2021{\natexlab{b}})\citenamefont {Cheng}, \citenamefont {Liu},
  \citenamefont {Baker},\ and\ \citenamefont {Hall}}]{Hall21}%
  \BibitemOpen
  \bibfield  {author} {\bibinfo {author} {\bibfnamefont {Shuming}\ \bibnamefont
  {Cheng}}, \bibinfo {author} {\bibfnamefont {Lijun}\ \bibnamefont {Liu}},
  \bibinfo {author} {\bibfnamefont {Travis~J.}\ \bibnamefont {Baker}}, \ and\
  \bibinfo {author} {\bibfnamefont {Michael J.~W.}\ \bibnamefont {Hall}},\
  }\href@noop {} {\enquote {\bibinfo {title} {Limitations on sharing {B}ell
  nonlocality between sequential pairs of observers},}\ } (\bibinfo {year}
  {2021}{\natexlab{b}}),\ \Eprint {http://arxiv.org/abs/2102.11574}
  {arXiv:2102.11574 [quant-ph]} \BibitemShut {NoStop}%
\bibitem [{\citenamefont {Bera}\ \emph {et~al.}(2018)\citenamefont {Bera},
  \citenamefont {Mal}, \citenamefont {Sen(De)},\ and\ \citenamefont
  {Sen}}]{Bera18}%
  \BibitemOpen
  \bibfield  {author} {\bibinfo {author} {\bibfnamefont {Anindita}\
  \bibnamefont {Bera}}, \bibinfo {author} {\bibfnamefont {Shiladitya}\
  \bibnamefont {Mal}}, \bibinfo {author} {\bibfnamefont {Aditi}\ \bibnamefont
  {Sen(De)}}, \ and\ \bibinfo {author} {\bibfnamefont {Ujjwal}\ \bibnamefont
  {Sen}},\ }\bibfield  {title} {\enquote {\bibinfo {title} {Witnessing
  bipartite entanglement sequentially by multiple observers},}\ }\href
  {\doibase 10.1103/PhysRevA.98.062304} {\bibfield  {journal} {\bibinfo
  {journal} {Phys. Rev. A}\ }\textbf {\bibinfo {volume} {98}},\ \bibinfo
  {pages} {062304} (\bibinfo {year} {2018})}\BibitemShut {NoStop}%
\bibitem [{\citenamefont {Maity}\ \emph {et~al.}(2020)\citenamefont {Maity},
  \citenamefont {Das}, \citenamefont {Ghosal}, \citenamefont {Roy},\ and\
  \citenamefont {Majumdar}}]{Maity20}%
  \BibitemOpen
  \bibfield  {author} {\bibinfo {author} {\bibfnamefont {Ananda~G.}\
  \bibnamefont {Maity}}, \bibinfo {author} {\bibfnamefont {Debarshi}\
  \bibnamefont {Das}}, \bibinfo {author} {\bibfnamefont {Arkaprabha}\
  \bibnamefont {Ghosal}}, \bibinfo {author} {\bibfnamefont {Arup}\ \bibnamefont
  {Roy}}, \ and\ \bibinfo {author} {\bibfnamefont {A.~S.}\ \bibnamefont
  {Majumdar}},\ }\bibfield  {title} {\enquote {\bibinfo {title} {Detection of
  genuine tripartite entanglement by multiple sequential observers},}\ }\href
  {\doibase 10.1103/PhysRevA.101.042340} {\bibfield  {journal} {\bibinfo
  {journal} {Phys. Rev. A}\ }\textbf {\bibinfo {volume} {101}},\ \bibinfo
  {pages} {042340} (\bibinfo {year} {2020})}\BibitemShut {NoStop}%
\bibitem [{\citenamefont {Srivastava}\ \emph {et~al.}(2021)\citenamefont
  {Srivastava}, \citenamefont {Mal}, \citenamefont {Sen(De)},\ and\
  \citenamefont {Sen}}]{sriv21}%
  \BibitemOpen
  \bibfield  {author} {\bibinfo {author} {\bibfnamefont {Chirag}\ \bibnamefont
  {Srivastava}}, \bibinfo {author} {\bibfnamefont {Shiladitya}\ \bibnamefont
  {Mal}}, \bibinfo {author} {\bibfnamefont {Aditi}\ \bibnamefont {Sen(De)}}, \
  and\ \bibinfo {author} {\bibfnamefont {Ujjwal}\ \bibnamefont {Sen}},\
  }\bibfield  {title} {\enquote {\bibinfo {title} {Sequential
  measurement-device-independent entanglement detection by multiple
  observers},}\ }\href {\doibase 10.1103/PhysRevA.103.032408} {\bibfield
  {journal} {\bibinfo  {journal} {Phys. Rev. A}\ }\textbf {\bibinfo {volume}
  {103}},\ \bibinfo {pages} {032408} (\bibinfo {year} {2021})}\BibitemShut
  {NoStop}%
\bibitem [{\citenamefont {Peres}(1996)}]{Peres96}%
  \BibitemOpen
  \bibfield  {author} {\bibinfo {author} {\bibfnamefont {Asher}\ \bibnamefont
  {Peres}},\ }\bibfield  {title} {\enquote {\bibinfo {title} {Separability
  criterion for density matrices},}\ }\href {\doibase
  10.1103/PhysRevLett.77.1413} {\bibfield  {journal} {\bibinfo  {journal}
  {Phys. Rev. Lett.}\ }\textbf {\bibinfo {volume} {77}},\ \bibinfo {pages}
  {1413--1415} (\bibinfo {year} {1996})}\BibitemShut {NoStop}%
\bibitem [{\citenamefont {Horodecki}\ \emph {et~al.}(1996)\citenamefont
  {Horodecki}, \citenamefont {Horodecki},\ and\ \citenamefont
  {Horodecki}}]{Hor96}%
  \BibitemOpen
  \bibfield  {author} {\bibinfo {author} {\bibfnamefont {Michał}\ \bibnamefont
  {Horodecki}}, \bibinfo {author} {\bibfnamefont {Paweł}\ \bibnamefont
  {Horodecki}}, \ and\ \bibinfo {author} {\bibfnamefont {Ryszard}\ \bibnamefont
  {Horodecki}},\ }\bibfield  {title} {\enquote {\bibinfo {title} {Separability
  of mixed states: necessary and sufficient conditions},}\ }\href {\doibase
  https://doi.org/10.1016/S0375-9601(96)00706-2} {\bibfield  {journal}
  {\bibinfo  {journal} {Physics Letters A}\ }\textbf {\bibinfo {volume}
  {223}},\ \bibinfo {pages} {1--8} (\bibinfo {year} {1996})}\BibitemShut
  {NoStop}%
\bibitem [{\citenamefont {Horodecki}\ \emph {et~al.}(1995)\citenamefont
  {Horodecki}, \citenamefont {Horodecki},\ and\ \citenamefont
  {Horodecki}}]{horodecki95}%
  \BibitemOpen
  \bibfield  {author} {\bibinfo {author} {\bibfnamefont {R.}~\bibnamefont
  {Horodecki}}, \bibinfo {author} {\bibfnamefont {P.}~\bibnamefont
  {Horodecki}}, \ and\ \bibinfo {author} {\bibfnamefont {M.}~\bibnamefont
  {Horodecki}},\ }\bibfield  {title} {\enquote {\bibinfo {title} {Violating
  {B}ell inequality by mixed spin-12 states: necessary and sufficient
  condition},}\ }\href {\doibase https://doi.org/10.1016/0375-9601(95)00214-N}
  {\bibfield  {journal} {\bibinfo  {journal} {Physics Letters A}\ }\textbf
  {\bibinfo {volume} {200}},\ \bibinfo {pages} {340--344} (\bibinfo {year}
  {1995})}\BibitemShut {NoStop}%
\bibitem [{\citenamefont {Fine}(1982)}]{Fine82}%
  \BibitemOpen
  \bibfield  {author} {\bibinfo {author} {\bibfnamefont {Arthur}\ \bibnamefont
  {Fine}},\ }\bibfield  {title} {\enquote {\bibinfo {title} {Hidden variables,
  joint probability, and the {B}ell inequalities},}\ }\href {\doibase
  10.1103/PhysRevLett.48.291} {\bibfield  {journal} {\bibinfo  {journal} {Phys.
  Rev. Lett.}\ }\textbf {\bibinfo {volume} {48}},\ \bibinfo {pages} {291--295}
  (\bibinfo {year} {1982})}\BibitemShut {NoStop}%
\bibitem [{\citenamefont {Terhal}(2000)}]{terhal00}%
  \BibitemOpen
  \bibfield  {author} {\bibinfo {author} {\bibfnamefont {Barbara~M.}\
  \bibnamefont {Terhal}},\ }\bibfield  {title} {\enquote {\bibinfo {title}
  {Bell inequalities and the separability criterion},}\ }\href {\doibase
  https://doi.org/10.1016/S0375-9601(00)00401-1} {\bibfield  {journal}
  {\bibinfo  {journal} {Physics Letters A}\ }\textbf {\bibinfo {volume}
  {271}},\ \bibinfo {pages} {319--326} (\bibinfo {year} {2000})}\BibitemShut
  {NoStop}%
\bibitem [{\citenamefont {Chru{\'{s}}ci{\'{n}}ski}\ and\ \citenamefont
  {Sarbicki}(2014)}]{Sarbicki14}%
  \BibitemOpen
  \bibfield  {author} {\bibinfo {author} {\bibfnamefont {Dariusz}\ \bibnamefont
  {Chru{\'{s}}ci{\'{n}}ski}}\ and\ \bibinfo {author} {\bibfnamefont
  {Gniewomir}\ \bibnamefont {Sarbicki}},\ }\bibfield  {title} {\enquote
  {\bibinfo {title} {Entanglement witnesses: construction, analysis and
  classification},}\ }\href {\doibase 10.1088/1751-8113/47/48/483001} {\
  \textbf {\bibinfo {volume} {47}},\ \bibinfo {pages} {483001} (\bibinfo {year}
  {2014})}\BibitemShut {NoStop}%
\bibitem [{\citenamefont {Simmons}(1963)}]{Simmons63}%
  \BibitemOpen
  \bibfield  {author} {\bibinfo {author} {\bibfnamefont {G.~F.}\ \bibnamefont
  {Simmons}},\ }\href@noop {} {\emph {\bibinfo {title} {Introduction to
  Topology and Modern Analysis}}}\ (\bibinfo  {publisher} {McGraw-Hill},\
  \bibinfo {year} {1963})\BibitemShut {NoStop}%
\bibitem [{\citenamefont {Lax}(2002)}]{Lax02}%
  \BibitemOpen
  \bibfield  {author} {\bibinfo {author} {\bibfnamefont {P.~D.}\ \bibnamefont
  {Lax}},\ }\href@noop {} {\emph {\bibinfo {title} {Functional Analysis}}}\
  (\bibinfo  {publisher} {Wiley-Interscience},\ \bibinfo {year}
  {2002})\BibitemShut {NoStop}%
\bibitem [{\citenamefont {G\"uhne}\ \emph {et~al.}(2002)\citenamefont
  {G\"uhne}, \citenamefont {Hyllus}, \citenamefont {Bru\ss{}}, \citenamefont
  {Ekert}, \citenamefont {Lewenstein}, \citenamefont {Macchiavello},\ and\
  \citenamefont {Sanpera}}]{guhne02}%
  \BibitemOpen
  \bibfield  {author} {\bibinfo {author} {\bibfnamefont {O.}~\bibnamefont
  {G\"uhne}}, \bibinfo {author} {\bibfnamefont {P.}~\bibnamefont {Hyllus}},
  \bibinfo {author} {\bibfnamefont {D.}~\bibnamefont {Bru\ss{}}}, \bibinfo
  {author} {\bibfnamefont {A.}~\bibnamefont {Ekert}}, \bibinfo {author}
  {\bibfnamefont {M.}~\bibnamefont {Lewenstein}}, \bibinfo {author}
  {\bibfnamefont {C.}~\bibnamefont {Macchiavello}}, \ and\ \bibinfo {author}
  {\bibfnamefont {A.}~\bibnamefont {Sanpera}},\ }\bibfield  {title} {\enquote
  {\bibinfo {title} {Detection of entanglement with few local measurements},}\
  }\href {\doibase 10.1103/PhysRevA.66.062305} {\bibfield  {journal} {\bibinfo
  {journal} {Phys. Rev. A}\ }\textbf {\bibinfo {volume} {66}},\ \bibinfo
  {pages} {062305} (\bibinfo {year} {2002})}\BibitemShut {NoStop}%
\bibitem [{\citenamefont {Gühne}\ \emph {et~al.}(2003)\citenamefont {Gühne},
  \citenamefont {Hyllus}, \citenamefont {Bruss}, \citenamefont {Ekert},
  \citenamefont {Lewenstein}, \citenamefont {Macchiavello},\ and\ \citenamefont
  {Sanpera}}]{guhne03}%
  \BibitemOpen
  \bibfield  {author} {\bibinfo {author} {\bibfnamefont {O.}~\bibnamefont
  {Gühne}}, \bibinfo {author} {\bibfnamefont {P.}~\bibnamefont {Hyllus}},
  \bibinfo {author} {\bibfnamefont {D.}~\bibnamefont {Bruss}}, \bibinfo
  {author} {\bibfnamefont {A.}~\bibnamefont {Ekert}}, \bibinfo {author}
  {\bibfnamefont {M.}~\bibnamefont {Lewenstein}}, \bibinfo {author}
  {\bibfnamefont {C.}~\bibnamefont {Macchiavello}}, \ and\ \bibinfo {author}
  {\bibfnamefont {A.}~\bibnamefont {Sanpera}},\ }\bibfield  {title} {\enquote
  {\bibinfo {title} {Experimental detection of entanglement via witness
  operators and local measurements},}\ }\href {\doibase
  10.1080/09500340308234554} {\bibfield  {journal} {\bibinfo  {journal}
  {Journal of Modern Optics}\ }\textbf {\bibinfo {volume} {50}},\ \bibinfo
  {pages} {1079--1102} (\bibinfo {year} {2003})}\BibitemShut {NoStop}%
\end{thebibliography}%
\end{document}